\documentclass[reqno]{amsart}

\usepackage{graphicx,epsfig}
\usepackage{multirow}
\usepackage{algorithm}
\usepackage{algpseudocode}
 \usepackage{setspace}
\usepackage{amsmath,amssymb,latexsym, amsfonts, amscd, amsthm, xy}
\input{xy}
\xyoption{all}

\input{xy}
\xyoption{all}

\numberwithin{equation}{section}

\usepackage[usenames]{color}

\makeindex
=651pt \headheight = 12pt \headsep = 0pt

\theoremstyle{plain}

\newtheorem{theorem}{Theorem}[section]

\newtheorem{proposition}[theorem]{Proposition}
\newtheorem{corollary}[theorem]{Corollary}

\newtheorem{lemma}{Lemma}[section]

\theoremstyle{definition}

\newtheorem{remark}{Remark}[section]

\def\R{\mathbb{R}}
\def\cM{\mathcal{M}}

\def\M{\mathcal{M}}

\DeclareMathOperator*{\argmin}{argmin}

\title {A projection approach for multiple monotone regression}
\author{Lizhen Lin, Brian St. Thomas, Walter W. Piegorsch, James Scott and Carlos Carvalho}
\email{lizhen.lin@nd.edu}
\email{brian.st.thomas@duke.edu}
\email{piegorsch@math.arizona.edu}
\email{james.scott@mccombs.utexas.edu}
\email{carlos.carvalho@mccombs.utexas.edu}

\address{Department of Applied and Computational Mathematics and Statistics,
The University of Notre Dame, Notre Dame, IN}
\address{Department of Statistical Science, Duke University, Durham, NC.}
\address{BIO5 Institute and Department of Mathematics, The University of Arizona, Tucson, AZ}
\address{The University of Texas McCombs School of Business, Austin, TX}
\address{The University of Texas McCombs School of Business, Austin, TX}

\begin{document}
\maketitle

\begin{abstract}

Shape-constrained inference has wide applicability in bioassay, medicine, economics, risk assessment, and many other fields. Although there has been a large amount of work on  monotone-constrained univariate curve estimation,  multivariate shape-constrained problems are much more challenging, and fewer advances have been made in this direction. With a focus on monotone regression with multiple predictors, this current work  proposes a projection approach to estimate a multiple monotone regression function.  An initial unconstrained estimator -- such as a local polynomial estimator or spline estimator -- is  first obtained,  which is then projected onto the shape-constrained space. A shape-constrained estimate
is obtained by sequentially projecting an  (adjusted) initial estimator along each univariate direction.  Compared to the initial unconstrained estimator, the projection estimate results in a reduction of estimation error in terms of both $L^p$ ($p\geq 1$) distance and supremum distance.  We also derive the asymptotic distribution of the projection estimate. Simple computational algorithms are available for implementing the projection in both the unidimensional and higher dimensional cases. Our work  provides a simple recipe for practitioners to use in real applications, and is illustrated with a joint-action example from environmental toxicology.

\textbf{Keywords}:  Monotone regression with multiple predictors; Drug interaction; Environmental risk assessment; Projection

\end{abstract}
\maketitle

\section{Introduction} \label{sec:intro}

Shape-constrained (e.g. monotone constrained) statistical inference is applied to a variety of data-analytic problems.  In environmental toxicology, for instance, monotone constraints are imposed based on natural assumptions that the response of subjects exposed to certain chemical pollutants will not in general decrease with the increasing pollutant dose \cite{PiXi14}. Another common application can be found in disease screening,  where the probability of disease is assumed non-decreasing with increasing measurements of a pertinent biomarker \cite{biomaker, baker2000}.  Or in economics, the demand and supply curve is in general assumed to be monotone \cite{demand-supply}.   Motivated by this large variety of applications, a panoply of statistical approaches has been developed for estimating \emph{monotone curves}, i.e., one-dimensional monotone functions. Frequentist methods in general fall into three categories; the first involves kernel based approaches such as described by \cite{Muller2}, \cite{mammen} and \cite{dette2}. The second class of methods
models
the regression function as a linear span of a spline basis such as in \cite{ramsay88} and \cite{kong2}.  The third class of methods
is
based on \emph{isotonic regression} \cite{roberston88,barlow72}, recent developments of which can be found in \cite{Bhattacharya1, bhli10, bhli11, bhli13} and \cite{SJOSlin}.  In addition, a few Bayesian approaches have been proposed in, for example, \cite{Bornkamp09}, \cite{Lin23022014}, \cite{shiv09} and \cite{shiv11}. Although there is large body of work on estimating monotone curves,  shape-constrained problems with respect to multiple predictors, which are,
e.g.,
important in drug interaction studies or in risk analyses involving the joint action of multiple pollutants,  are more challenging.  This is due to the difficulty in incorporating the multivariate shape constraints.  Along these lines, \cite{Saarela} proposes a Bayesian approach for multiple regression  using marked point processes, while \cite{Lin23022014} combines Gaussian processes with projections for estimating a multivariate monotone function. In \cite{chetno09}, a monotone arrangement procedure \cite{inequality} is applied to an initial unconstrained estimator.


Our motivation here is to develop a theoretically appealing, computationally feasible, and convenient-to-implement approach for  estimating monotone constrained functions with multiple predictors.
We propose use of a \emph{projection} of some initial, unconstrained estimator of the regression function, i.e., finding the monotone function closest to this initial estimator in some distance norm. Such estimates are intuitive, and can result in  reduction of the estimation error compared to that of a na\"{i}ve initial estimator.  Note that a general  projection framework for constrained functional parameters, in particular for constrained functions forming a closed convex cone  in a Hilbert space, is proposed in \cite{frechref}. Our approach falls into this general framework; however, our projection algorithm makes use of the fact that the convex  cone of a multivariate monotone function is an intersection of a collection of convex cones of univariate monotone functions.  We then derive an expression  for the projection estimate, making use of unidimensional projections.  It can be  shown that the projection  of an initial estimator onto the space of monotone functions with multiple predictors can be decomposed into sequential projections of an adjusted initial estimator along each univariate direction.  This simplifies the operation substantially, and allows us to suggest
computational algorithms for approximating such functions.

In the next section, we study in detail this monotone projection framework and the properties of our projection estimates.
In section \ref{sec:newSec3}, we describe a bootstrap methodology for constructing confidence intervals. In section \ref{sec:newSec4},
we carry out a simulation study to explore the methods' operating characteristics, and
we apply our methods to a two-dose, joint-action data set from environmental toxicology.
Section \ref{sec:newSec5} ends with a short discussion.

\section{A projection framework for monotone regression with multiple predictors} \label{sec:Sec2}

\subsection{Preliminary estimator for the proposed approach}

Let $x\in \R^p$ be a $p$-dimensional predictor and $y$ be a response variable.
The response variable can be
discrete or continuous, depending on the application of interest.
We define the regression function $F(x)$ in a general framework as
\begin{align}
\label{eq-model1}
F(x)=E(y \mid x).
\end{align}
For instance, if $y$ is binary, taking values 0 or 1, then take $F(x)$ as the response probability $F(x)=P(y=1\mid x)$.

Denote the data as $(x_i,y_i)$, $i=1,\ldots, n$. Without loss of generality,  we assume the predictors or covariates
satisfy
$x_i=(x_{i1},\ldots, x_{ip})\in [0,1]^p\subset\R^p$.   The regression function $F(x)$ is assumed to be \emph{monotone} with respect to a natural partial ordering on $\R^p$, that is,
for $x_1=(x_{11},\ldots, x_{1p})$,
$x_2=(  x_{21},\ldots, x_{2p})$ and $x_{1j}\leq x_{2j}$ (for $j=1,\ldots, p$), one has $F(x_1)\leq F(x_2)$. We are interested in conducting inference on $F(x)$ under the monotonicity constraint.

Denote $\M$ as the space of monotone functions on $\mathcal{X}=[0,1]^p$. It can be seen that $\M$ is a closed convex cone. Our approach relies on \emph{projecting} an initial estimator of $F(x)$  on to $\M$. For instance, the initial estimator could be a local polynomial estimator such as the popular kernel estimator or a local linear estimator. Alternatively,  one could employ expansions of spline bases such as a $B$-spline basis.  In any case, we will show the resulting projection estimates exhibit desirable theoretical properties as well as good finite-sample performance. Efficient computational algorithms  are also straightforward to develop for  implementing our approach.

To illustrate, we first consider a local polynomial regression estimator for a one-dimensional curve. Let $K(x)$ be a kernel function with $\int K(x)dx=1$, $\int xK(x)dx=0$  and $\int x^2K(x)dx < \infty$.  Denote $K_h(x)=h^{-1}K(x/h)$.  Let $\boldsymbol{\gamma}=(\gamma_0,\gamma_1,\ldots, \gamma_p)$ be a $(p+1)$-dimensional coefficient vector, and take
\begin{align*}
\widehat{\boldsymbol{\gamma}}=\arg \min_{\boldsymbol{\gamma}\in \R^{p+1}} \sum_{i=1}^n K_h(x-x_i)\left\{y_i - \sum_{j=0}^p\gamma_j(x-x_i)^j\right\}^2.
\end{align*}
An initial local polynomial estimator of $F(x)$ can be simply
\begin{align*}
\widehat{F}(x)=\widehat{\gamma}_0,
\end{align*}
which can be fitted  easily  via weighted least squares.

Another popular class of methods for nonparametric regression involves spline models \cite{Zhang2012}. More precisely,  one can construct a class of initial estimators by modeling the regression function as a linear span of  a spline basis, the most of popular of which is the $B$-spline basis \cite{EilerMar10}.
Given the knot sequence $\tau_1,\tau_2,\cdots,\tau_N$, the cubic B-spline basis $\{B_{j,4}\}_{j=1}^4$  is defined recursively as follows:
\begin{align*}
&B_{j,1}(x)=
\begin{cases}
1 & \tau_{j}\leq x\leq\tau_{j+1}\\
0 & \text{otherwise,}
\end{cases}\\
&B_{j,l}(x)=\frac{x-\tau_j}{\tau_{j+l-1}-\tau_j}B_{j,l-1}(x)+\frac{\tau_{j+l}-x}{\tau_{j+l}-\tau_{j+1}}B_{j+1,l-1}(x)\;\;(l=2,3,4).
\end{align*}
With this, one 
obtains an initial estimator of  $F(x)$ as a linear combination of the spline basis functions.  Using the $B$-spline basis above, this sets
\begin{equation}
\label{spline F}
F(x)=\sum_{j=1}^{N}\beta_j B_{j,4}(x),
\end{equation}
where $(\beta_1,\ldots, \beta_N)$ is the vector of coefficients.
The estimate is $\widehat F(x)=\sum_{i=1}^N\widehat{\beta}_iB_i(x)$ with
 \begin{align*}
(\widehat\beta_1,\ldots, \widehat\beta_N)=\arg\min_{\boldsymbol{\beta}\in \R^{m}} \left[\lambda \sum_{i=1}^n\{y_i-F(x_i)\}^2+(1-\lambda)\int \{F''(x)\}^2dx\right].
 \end{align*}
Here $\lambda$ is a smoothing parameter that controls the tradeoff between tighter fit to the data and smoothness of the estimates.
We refer to \cite{gracebook} for a general reference on smoothing splines.

For the multivariate case, one can easily obtain a kernel based initial estimator by employing a multivariate kernel $K$.  Or one can obtain an initial estimator using tensor product B-splines \cite{EilerMar10}. For example, take $p=2$ with the two-dimensional function $F(x_1, x_2)$. Let $\{B_{i1}(x_1)\}$, $i=1,\ldots, N_1$, be a B-spline basis along the $x_1$ direction, and $\{B_{j2}(x_2)\}$, $j=1,\ldots, N_2$, be a spline basis along the $x_2$ direction. A tensor product spline basis is given by $\{ B_{i1}B_{j2} \}$, $i=1,\ldots, N_1$, $j=1,\ldots, N_2$.  The multivariate function is modeled as a linear span of the tensor product  B-spline basis.   An initial estimator $\widehat{F}(x_1, x_2)$  can be obtained by minimizing the objective function
\begin{align*}
\lambda \sum_{i=1}^n\{y_i-F(x_{i1}, x_{i2})\}^2+(1-\lambda)\int\int \left\{    \left( \dfrac{\partial^2 F}{\partial x_1^2}\right)^2+2\left( \dfrac{\partial^2 F}{\partial x_1x_2}\right )^2+\left(\dfrac{\partial^2 F}{\partial x_2^2}\right)^2\right\}dx_1dx_2.
\end{align*}

With any initial estimator $\widehat{F}(x_1,x_2)$, one  then \emph{projects} $\widehat{F}(x_1,x_2)$ onto  the monotone space $\M$, which produces our ultimate estimator  of $ F(x_1,x_2)$.  We now proceed to give a rigorous definition of this projection and characterize its properties in the next subsection.

\subsection{A projection framework for shape constrained estimators}

Let $w$ be a function on $\mathcal{X}=[0,1]^p$.  We define the projection of $w$ onto the constrained space $\mathcal{M}$ as
\begin{equation}
\label{eq-mother}
 P_w =\argmin_{G \in \mathcal{M}} \int_{\mathcal{X}} \{ w(t) - G(t) \}^2 dt.
\end{equation}
That is, the projection estimate is defined to be the element in $\mathcal{M}$ that is closest to the initial (pre-projected function) estimator in $L^2$ distance.

Recall $\mathcal{M}=\mathcal{M}[0,1]^p$, which is  the space of monotone functions on $[0,1]^p$. Focusing initially on the $p=1$ case, \eqref{eq-mother} has the following closed form solution (see \cite{Lin23022014}  and \cite{Anevksi:2011})
\begin{equation}
\label{projection}
P_w(x)=\inf_{v\geq x} \sup_{u\leq x} \dfrac{1}{v-u}\int_u^vw(t)dt,\; \;\; x\in[0,1].
\end{equation}
The existence and uniqueness of the projection follow from Theorem 1 in \cite{rychlik}.

Letting $h$ be any function on [0,1], the \emph{greatest convex minorant}  of  $h$ is defined by
\begin{equation}
\label{eq-gcv}
T(h)=\arg\max\{z: z\leq h, \;z \;\text{convex} \}
\end{equation}
The solution \eqref{projection} is the slope of the greatest convex minorant of $\overline{w}(t)=\int_0^t w(s)ds$.

\begin{remark}
The projection in \eqref{projection} can be well approximated using the  pooled adjacent violators algorithm \cite{barlow72}. 
\end{remark}

One can easily generalize the above one-dimensional projection  algorithm to multiple dimensions ($p>1$). Take $p=2$; for which the following algorithm \cite{Lin23022014} converges to the two-dimensional projection
$$P_w=\argmin_{G\in\M} \int_0^1\int_0^1\{w(s,t)-G(s,t)\}^2ds dt.$$

\textbf{Algorithm 1}. 
For any fixed $t$, $w(s,t)$ is a function of $s$ and we use the projection \eqref{projection} to obtain a monotone function in $s$. We perform this projection for all values of $t$, and denote the resulting surface by $\widehat{w}^{(1)}(s,t)$.
Letting $S^{(1)} = \widehat{w}^{(1)} - w$, for any fixed $s$, project $w + S^{(1)}$ as a function of $t$ onto $\mathcal{M}[0,1]$ using \eqref{projection}. Perform this projection for all values of $s$ and denote the surface by $\widetilde{w}^{(1)}(s,t)$.
Set $T^{(1)} = \widetilde{w}^{(1)} - (w+S^{(1)})$.
Letting $i=2,\ldots,k$, in the $i$th step  we obtain $\widehat{w}^{(i)}$ by projecting $w + T^{(i-1)}$ along the $s$ direction for every fixed $t$ value in [0,1] and $\widetilde{w}^{(i)}$ as the projection of $w+S^{(i)}$ along the $t$ direction for every fixed $s$ value in [0,1].
 The algorithm terminates when $\widehat{w}^{(i)}$ or $\widetilde{w}^{(i)}$ is monotone with respect to both $s$ and $t$ for some step $i$.

Via an induction argument, one can show that projecting a $p$-dimensional function onto $\mathcal M[0,1]^p$ with $p>2$ can be characterized similarly to Algorithm 1, above, by introducing $p$ residual sequences.

Given an initial estimate $\widehat{F}(x)$, we denote the \emph{projection estimate} of $F(x)$ under the shape constraints  as $\widetilde{F}(x)$ with
\begin{equation}
\widetilde{F}(x)=P_{\widehat{F}(x)}(x).
\end{equation}
$\widetilde{F}(x)$ is the ultimate estimate used for inference.

Now, let \begin{align*}
\|\widetilde{F}(x)-F(x)\|_2&=\left[\int_{[0,1]^p} \{\widetilde{F}(x)-F(x)\}^2dx\right]^{1/2}\\
&=\left[\int_0^1\cdots\int_0^1\{\widetilde{F}(x)-F(x)\}^2dx_1\cdots dx_p\right]^{1/2},
\end{align*}
which is the $L^2$ norm for a $p$-dimensional function $F(x)$.

The following propositions show that $\widetilde{F}(x)$ is `closer' to the true regression monotone function $F(x)$, compared to the initial estimator $\widehat{F}(x)$. As a consequence, $\widehat{F}(x)$ produces smaller error in the $L^2$ norm compared to that of the initial estimator $\widehat{F}(x)$.
\begin{proposition}
\label{prop-2.1}
Let $x\in\R^p$ ($p\geq 1$). Let $\widehat{F}(x)$ be an initial estimator of $F(x)$ and $\widetilde{F}(x)=P_{\widehat{F}(x)}.$ Then the following holds:
\begin{align}
\|\widetilde{F}(x)-F(x)\|_2\leq \|\widehat{F}(x)-F(x)\|_2
\end{align}
\end{proposition}

For a proof, see the Appendix.

Corollary \ref{coro-1} follows immediately from  Proposition  \ref{prop-2.1}.

\begin{corollary}
\label{coro-1}
\begin{align*}
\|\widehat{F}(x)-F(x)\|_2=O(\lambda_n).
\end{align*}
where $\lambda_n\rightarrow 0$ as $n\rightarrow \infty$. Then
\begin{align*}
\|\widetilde{F}(x)-F(x)\|_2=O(\lambda_n).
\end{align*}
\end{corollary}

In the one-dimensional case, we can derive more general results on the reduction in estimation error for the projection estimator.
\begin{theorem}
\label{th-1d}
Assume $x\in \R$. Let $\Phi$ be any convex function and $\widetilde{F}(x)$ be an initial estimator of $F(x)$, such as a kernel estimator or a spline estimator. Let  $\widetilde{F}(x)=P_{\widehat{F}(x)}(x).$ One has
\begin{equation}
\label{eq-phip}
\int_0^1\Phi\left\{\widetilde{F}(x)-F(x)\right\}dx\leq \int_0^1\Phi\left\{\widehat{F}(x)-F(x)\right\}dx.
\end{equation}
\end{theorem}

For a proof, see the Appendix.

\begin{corollary}
\label{coro-1d}
Assume $x\in \R$. Let $\widetilde{F}(x)$ be an initial estimator of $F(x)$ such as a kernel estimator or a spline estimator. Let  $\widetilde{F}(x)=P_{\widehat{F}(x)}(x).$ One has
\begin{equation}
\label{eq-sup}
\sup_x|\widetilde{F}(x)-F(x)|\leq \sup_x|\widehat{F}(x)-F(x)|,
\end{equation}
and
\begin{equation}
\label{eq-lp}
\int_0^1|\widetilde{F}(x)-F(x)|^qdx\leq \int_0^1|\widehat{F}(x)-F(x)|^qdx,
\end{equation}
where $q\in[1,\infty)$.
\end{corollary}

\begin{proof}
Taking $\Phi=\|x\|^q$ in Theorem \ref{th-1d}, \eqref{eq-lp} follows. As $q\rightarrow \infty$, \eqref{eq-sup} holds.
Note that a different proof for \eqref{eq-sup} is given in \cite{Lin23022014}.
\end{proof}


The following proposition concerns the asymptotic distribution of the projection, which follows from the general results of Theorem 3.4 in \cite{frechref}.
\begin{proposition}
Let the tangent cone of ${\mathcal M}$ be $T_{\mathcal M}(F)$, which is  the closure of $\{\lambda(G-F), G\in \mathcal M, \lambda>0\}$.
If
\begin{align*}
\dfrac{\widehat{F}(x)-F(x)}{t_n}\xrightarrow{\mathcal{L}} U(x)
\end{align*} for some $t_n\rightarrow 0$ as $n\rightarrow \infty$,
then
\begin{align}
\dfrac{\widetilde{F}(x)-F(x)}{t_n}\xrightarrow{\mathcal L} \widetilde P_U(x),
\end{align}
where $\widetilde P_U(x)$ is the projection of $U(x)$ onto the tangent cone   $T_{\mathcal M}(F)$ and $\xrightarrow{\mathcal L}$ indicates convergence in distribution.
\end{proposition}


\begin{remark}
Let $\mathcal M_1,\ldots, \mathcal M_p$ be $p$ convex cones of functions such that $\mathcal{M}_k$ is the convex cone of functions which are monotone with respect to the $k$th direction of $x$.  That is, for any $F\in \mathcal M_k$,  $F(x_1,\ldots, x_p)$ is monotone with respect to $x_k$ at any fixed value of the other coordinates.  It is not difficult to see that $\mathcal M$ is the intersection of the $p$ convex cones  $\mathcal M_1,\ldots, \mathcal M_p$, i.e., $\mathcal M=\cap_{k=1}^p\mathcal M_k$.  The projection algorithm (Algorithm 1) can be viewed as a sequential projection of a $p$-variate function onto each convex cone $\mathcal M_k$,  while at each step the projected function is adjusted by adding residual sequences from the previous step.  Note that similar algorithms consisting of a projection step and an adjustment step hold for any space that can be written as the intersection of a collection of  convex cones (see \cite{general-proj}).

\end{remark}

%
%
%

\section{Bootstrap confidence intervals}
\label{sec:newSec3}

In this section we appeal to the bootstrap \cite{DaHi97} for constructing confidence intervals from our estimators.
Consider the following procedure for generating a nonparametric estimate, along with confidence intervals, for a monotone function $F(x)$.  The data are $(x_i, y_i)$ pairs, $i=1, \ldots, n$.
\begin{enumerate}
\item Fit the curve to obtain an initial estimator $\widehat{F}$  which need not satisfy the monotonicity constraint(s).
\item Project the estimate onto $\mathcal{M}$, the space of monotone functions.  Call this projected estimate $\widetilde F$, and let $e_i = y_i - \widetilde F(x_i)$ be the residual of the $i$th point from the monotone estimate.
\item Repeat the following procedure  $B =2000$ times.
\begin{enumerate}
\item Resample the residuals with replacement, yielding $e^\star_1, \ldots, e^\star_n$.
\item Set $y^\star_i = \widetilde F (x_i) + e^\star_i$.
\item Apply steps 1 and 2 to the bootstrapped points $(x_i, y_i^\star)$ to yield $\widetilde F^\star$.
\end{enumerate}
\item Collect the $B$ different bootstrapped estimates $\widetilde F^\star$.  We adopt percentile-based methods for constructing point-wise confidence intervals \cite{DiRo88}.  For example, the lower and upper
 95\% confidence bounds  of $F(x)$ are given by the 2.5 percentile and 97.5 percentile  of the ranked bootstrapped estimates, respectively.
 \end{enumerate}

In the next section,
both a simulation study and a real data example are presented to illustrate use of the bootstrap procedure
for constructing 95\%  pointwise  confidence intervals on the monotone regression function in Section \ref{sec-simu}.

\section{Simulation studies and data analysis} \label{sec:newSec4}
\label{sec-simu}
In this section we report the results of a series of short Monte Carlo simulations used to gauge selected operating characteristics of the projection estimator.  We also illustrate the methodology using a contemporary data set from environmental toxicology.

\subsection{Root mean squared errors}
\label{sec:simuRMSE}


To study the various features of our projection estimators, we simulated data from a variety of one-, two-, and three-predictor monotone regression models under a normal parent distribution. For example, with one predictor, $x_1$, we generated $y_i \sim \text{N}\big\{F(x_{i1}),\sigma^2\big\}, \ i=1,\ldots,n$ with $n$ set to $100$, where the covariates were taken to be equidistant on their domain,  and the mean function $F(x_1)$ was chosen from a class of monotone curves originally proposed by \cite{holmes} and \cite{Neelon}:
\begin{itemize}
  \item [(a)] flat function, $F_{11}(x_1)=3$, $x_1\in(0,10]$;
  \item [(b)] sinusoidal function, $F_{12}(x_1)=0.32\{x_1+\sin(x_1)\}$, $x_1\in(0,10]$;
  \item [(c)] step function, $F_{13}(x_1)=3$ if $x_1\in(0,8]$ and $F_3(x_1)=6$ if $x_1\in (8,10]$;
  \item [(d)] linear function, $F_{14}(x_1)=0.3x_1$, $x_1\in(0,10]$;
  \item [(e)] exponential function, $F_{15}(x_1)=0.15\exp\{0.6x_1-3\}$, $x_1\in(0,10]$;
  \item [(f)] logistic function, $F_{16}(x_1)=3/\left(1+\exp\{-2x_1+10\}\right)$, $x_1\in(0,10]$.
  \item[(g)] half-normal function, $F_{17}(x_1)=3\exp\left\{-\frac12(0.02)^2(0.1x_1-1)^2\right\}$,  $x_1\in(0,10]$.
  \item[(h)] mixture function, $F_{18}(x_1)=6F(0.1x_1)$, where $F(\cdot)$ is the c.d.f.~from the equal mixture of $N(0.25, 0.004^2)$ and $N(0.75, 0.04^2)$.
\end{itemize}
The standard deviation parameter $\sigma$ was varied over a range of $\sigma = 0$ (i.e., a purely deterministic response)$,0.1, 0.2, \ldots, 1.2, 1.3$.  These same models were also employed at $n=100$ and $\sigma=1$ in a comparative study by \cite{shiv09}.

Each simulation was replicated 50 times for every value of $\sigma$, and the square root of the empirical root mean squared error (RMSE) for the projection estimator was recorded using either a kernel estimator or a cubic spline estimator for the initial $\widehat{F}(x)$.  More precisely, the RMSE is defined as
\begin{align*}
\text{RMSE}(\widehat{F}(x))=\sqrt{\frac{1}{n}\sum_{i=1}^n \{\widehat{F}(x)-F(x)\}^2}.
\end{align*}
These resultant RMSEs using either initial estimator were then averaged over the 50 replicate simulation trials for each model.  The results are summarized in Figure \ref{fig-1d} for the initial kernel estimator (top) and the initial spline estimator (bottom).  One sees that the average RMSE rises gradually but consistently with increasing $\sigma$, and that these errors are substantially higher with both the step function, $F_{13}(x)$, and the mixture function, $F_{18}(x)$, models.  By contrast, the flat function, $F_{11}(x)$, and half-normal function, $F_{17}(x)$, models separate out with lower average RMSEs as $\sigma$ rises.  Also, for many values of $\sigma$, the average spline-based RMSE exceeds the average kernel-based RMSE.
Table \ref{table-1d} gives further detail on the one-predictor RMSE simulations, by reporting the average RMSEs and the standard error of the  RMSES at $\sigma = \frac12$ and $\sigma = 1$ across all models. Note that the standard errors of the 50 RMSE values are quite small, indicating that there is small variability for the average RMSE.


\begin{figure}[!h]
\caption{Average root mean squared error (RMSE) for one-dimensional projection estimates based on initial kernel estimator (top) and initial spline estimator (bottom) for $\widetilde{F}(x)$. Digits correspond to model numbers (second digits) in Sec.~\ref{sec:simuRMSE}.}
\includegraphics[width=9cm]{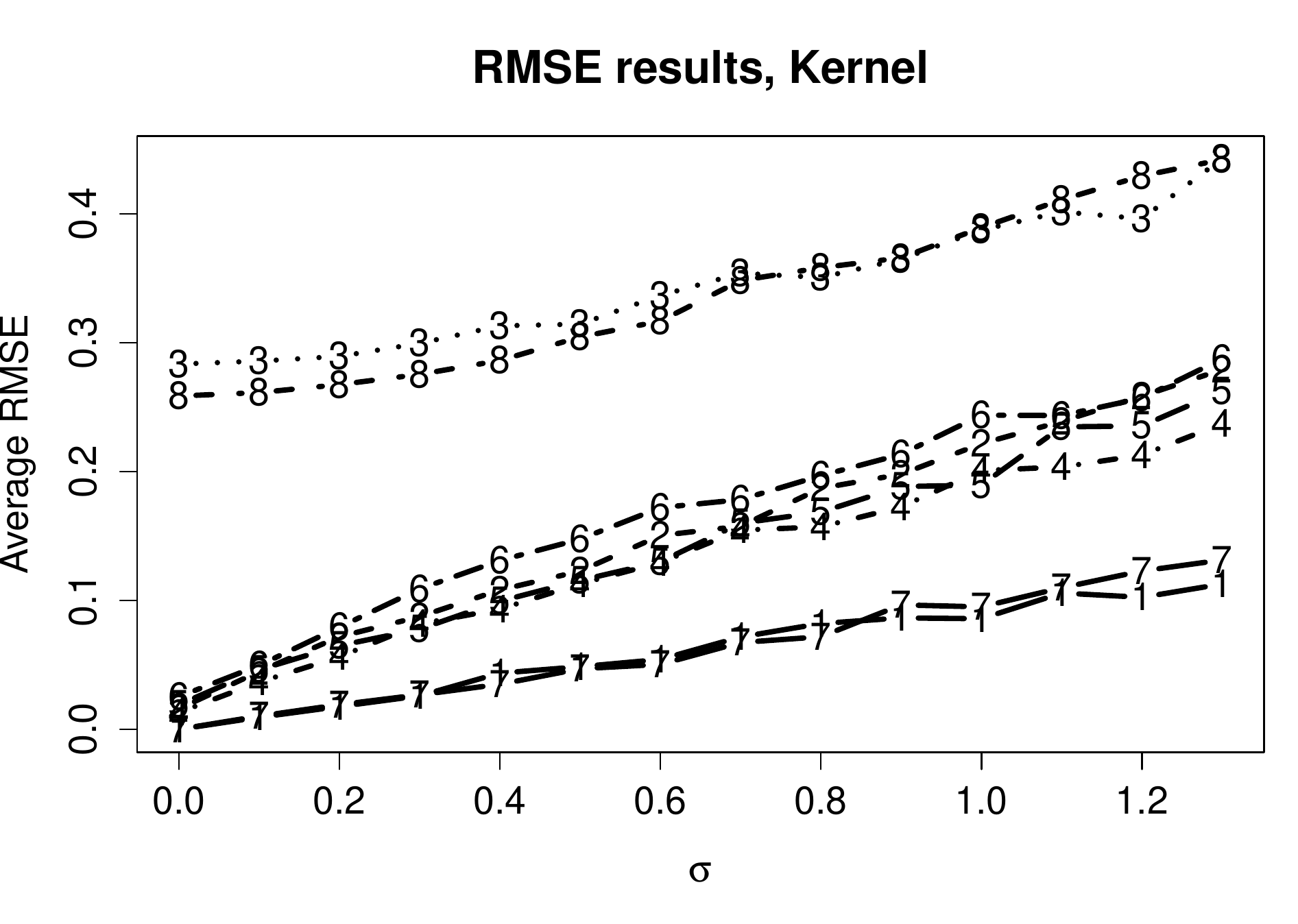}
\includegraphics[width=9cm]{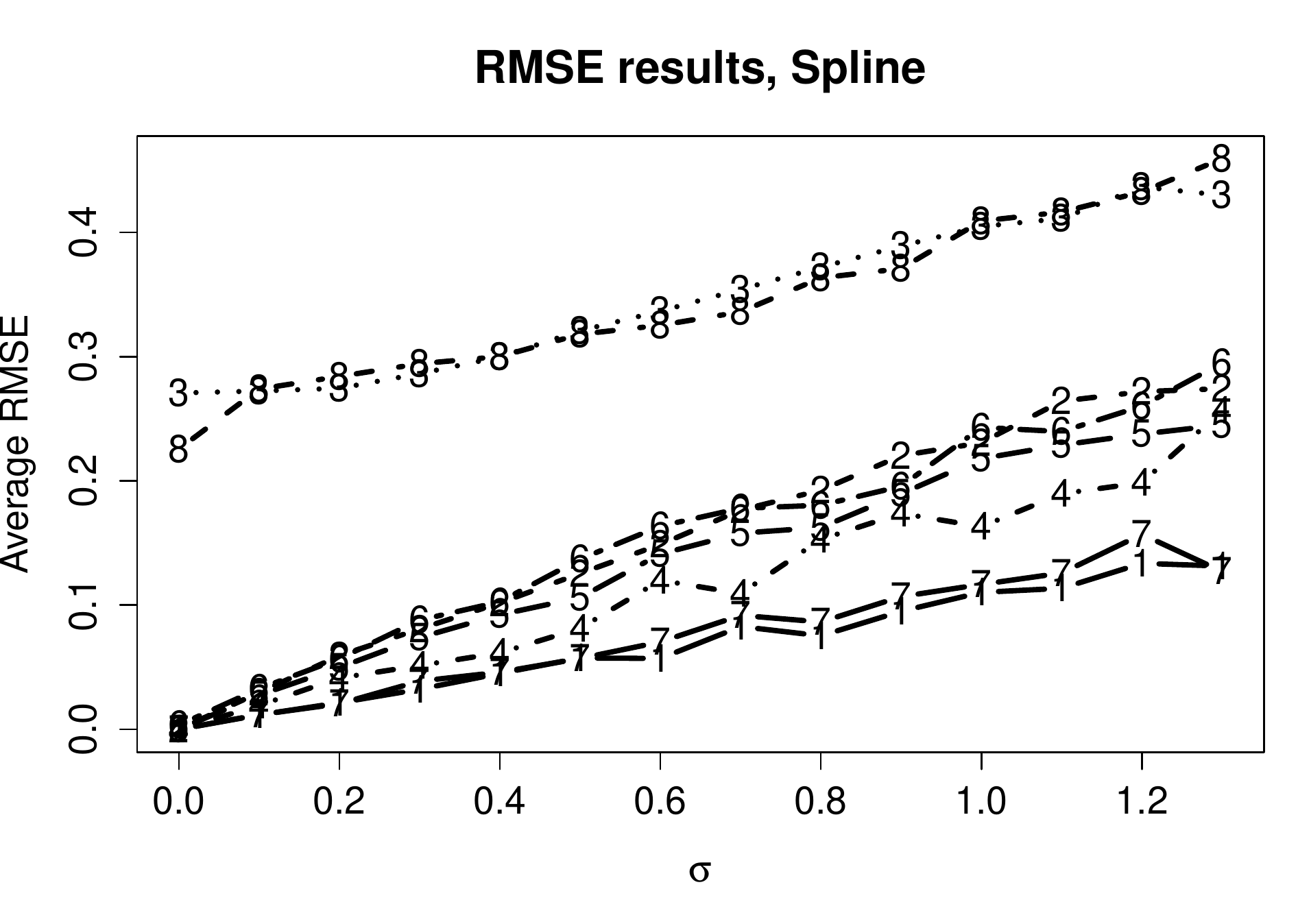}
\label{fig-1d}
\end{figure}

\begin{table}[!h]
\label{table-1d}
%
\caption{The mean and standard error of one-predictor root mean square errors (RMSEs) for simulated data at $n=100$, averaged across 50 simulation replicates, listed by underlying mean function $F(x_1)$ from Sec.~\ref{sec:simuRMSE} and standard deviation parameter $\sigma$.}
\begin{tabular}{lllcccc}
Mean function, $F(x_1)$  & Initial estimator & & $\sigma=\frac12$ &   $\sigma=\frac12$  & $\sigma=1$ &$\sigma=1$\\
\hline
                                                 &                           &&            Average                  &    SE        &  Average         &  SE    \\
\hline
flat, $F_{11}(x_1)$        & Kernel            & & 0.0485    &  0.0048      & 0.0858 & 0.0101  \\
                           & Spline            & & 0.0573        &  0.0053 & 0.1102  & 0.0150 \\ \\
sinusoidal, $F_{12}(x_1)$  & Kernel            & & 0.1228  & 0.0049    & 0.2217  & 0.0080\\
                           & Spline            & & 0.1255         & 0.0063  & 0.2307& 0.0099  \\ \\
step, $F_{13}(x_1)$        & Kernel            & & 0.3144     & 0.0072     & 0.3875 & 0.0113 \\
                           & Spline            & & 0.3209  &   0.0043      & 0.4048 & 0.0114 \\ \\
linear, $F_{14}(x_1)$      & Kernel            & & 0.1129    &  0.0053      & 0.2009 & 0.0106 \\
                           & Spline            & & 0.0813         & 0.0074 & 0.1631 & 0.0159 \\ \\
exponential, $F_{15}(x_1)$ & Kernel            & & 0.1155  & 0.0047        & 0.1895 & 0.0110 \\
                           & Spline            & & 0.1060           &0.0053 & 0.2179& 0.0159  \\ \\
logistic, $F_{16}(x_1)$    & Kernel            & & 0.1475          & 0.0051 & 0.2438 & 0.0103 \\
                           & Spline            & & 0.1369   & 0.0055       & 0.2430& 0.0111  \\ \\
half-normal, $F_{17}(x_1)$ & Kernel            & & 0.0470    &   0.0048    & 0.0950  & 0.0087\\
                           & Spline            & & 0.0570           &  0.0060 & 0.1166 & 0.0144 \\ \\
mixture, $F_{18}(x_1)$     & Kernel            & & 0.3045    & 0.0040      & 0.3893 & 0.0086 \\
                           & Spline            & & 0.3180           & 0.0042 & 0.4092 & 0.0086 \\
\hline
\end{tabular}
\end{table}



For the two-predictor setting with $x_1$ and $x_2$, we again generated $y_i \sim \text{N}\big\{F(x_{i1},x_{i2}),\sigma^2\big\}, \ i=1,\ldots,100$, with $\sigma$ ranging over $\sigma = 0,0.1,\ldots,1.3$. Now, $(x_1, x_2)\in [0,1]\times [0,1]$ and the two-dimensional mean functions were taken from a study considered in \cite{Saarela}:
 \begin{itemize}
  \item [(a)] $F_{21}(x_1, x_2)=\sqrt{x_1}$;
  \item [(b)] $F_{22(}x_1, x_2)=0.5x_1+0.5x_2$;
  \item [(c)] $F_{23}(x_1, x_2)=\min(x_1, x_2)$;
  \item [(d)] $F_{24}(x_1, x_2)=0.25x_1+0.25x_2+0.5\times 1_{\{x_1+x_2>1\}}$;
  \item [(e)] $F_{25}(x_1, x_2)=0.25x_1+0.25x_2+0.5\times 1_{\{\min(x_1, x_2)>0.5\}}$;
  \item [(f)] $F_{26}(x_1, x_2)=1_{\{ (x_1-1)^2+(x_2-1)^2<1  \} }\times \sqrt{1-(x_1-1)^2-(x_2-1)^2}.$
  \end{itemize}
The consequent average RMSEs are plotted in Figure \ref{fig-2d}, separated by initial kernel estimator (top) and initial spline estimator (bottom).   The broad patterns appear similar to those seen with one predictor, although there is less separation/distinguishablity in the RMSEs across models.
Also, the spline-based RMSEs are often much closer to their kernel-based cousins, although for large $\sigma$ the former generally still exceed the latter.  
Table \ref{mse-2d} reports the average RMSE values at $\sigma=\frac12$ and $\sigma=1$.

\begin{figure}[!h]
\caption{Average root mean squared error (RMSE) for two-dimensional projection estimates based on initial kernel estimator (top) and initial spline estimator (bottom) for $\widetilde{F}(x)$. Digits correspond to model numbers (second digits) in Sec.~\ref{sec:simuRMSE}.}
\includegraphics[width=9cm]{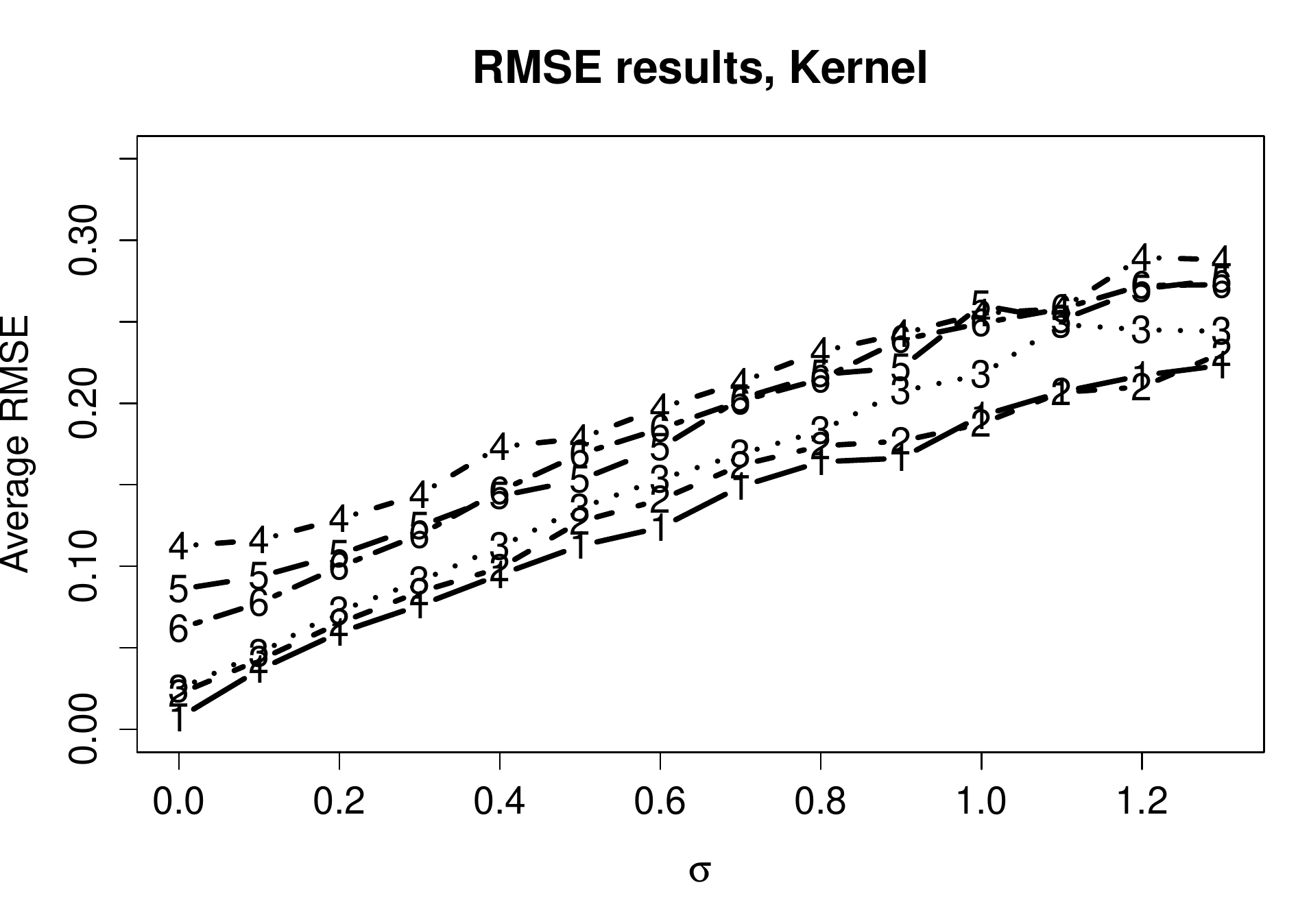}
\includegraphics[width=9cm]{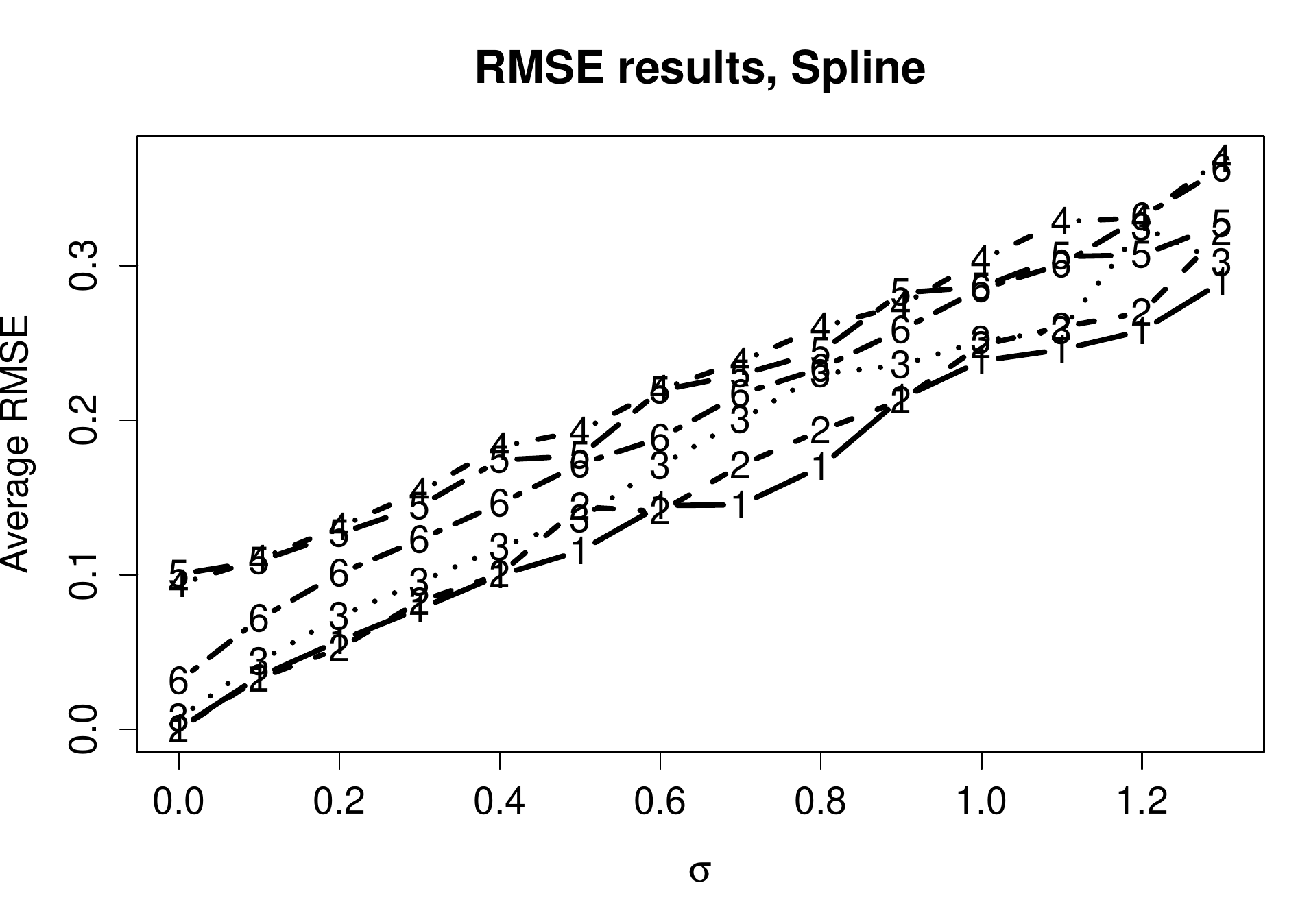}
\label{fig-2d}
\end{figure}
  \begin{table}[!h]
%
\caption{The mean and standard error of two-predictor root mean square errors (RMSEs) for simulated data at $n=100$, averaged across 50 simulation replicates, listed by underlying mean function $F(x_1, x_2)$ from Sec.~\ref{sec:simuRMSE} and standard deviation parameter $\sigma$.}
\begin{tabular}{lllcccc}
Mean function, $F(x_1, x_2)$  & Initial estimator & & $\sigma=\frac12$ &   $\sigma=\frac12$  & $\sigma=1$ &$\sigma=1$\\
\hline
                                                 &                           &&            Average                  &    SE       &  Average         &  SE   \\
\hline
$F_{21}(x_1, x_2)$            & Kernel            & & 0.1125   &      0.0074           & 0.1923 & 0.0158  \\
                                           & Spline            & & 0.1153           & 0.0075  &    0.2383  & 0.0138 \\ \\
$F_{22}(x_1, x_2)$            & Kernel            & & 0.1274    &  0.0054        & 0.1874 & 0.0137 \\
                              & Spline            & & 0.1439      &  0.0054 &      0.2484  & 0.0117 \\ \\
$F_{23}(x_1, x_2)$            & Kernel            & & 0.1356   &  0.0065       & 0.2175 & 0.0111  \\
                              & Spline            & & 0.1360      &  0.0058      & 0.2514  & 0.0220 \\ \\
$F_{24}(x_1, x_2)$            & Kernel            & & 0.1783  &  0.0071          & 0.2551 & 0.0130 \\
                              & Spline            & & 0.1929       & 0.0055  &     0.3036  & 0.0010 \\ \\
$F_{25}(x_1, x_2)$            & Kernel            & & 0.1530    &   0.0056       & 0.2601 &  0.0137\\
                              & Spline            & & 0.1768         &  0.0051 &    0.2861 & 0.0149 \\ \\
$F_{26}(x_1, x_2)$            & Kernel            & & 0.1687      &  0.0062       & 0.2491 & 0.0116  \\
                              & Spline            & & 0.1715     &  0.0055 &        0.2849  & 0.0116 \\
\hline
\end{tabular}
\label{mse-2d}
\end{table}

Lastly, we simulated $n = 100$ data points from models whose true regression functions involve three predictors, $x_1, x_2, x_3$, where each $x_j\in [0,10]$, and again using a Normal model with constant standard deviation $\sigma$ ranging over $\sigma = 0,0.1,\ldots,1.3$.  For the underlying mean monotone functions we employed the following collection:
       \begin{itemize}
  \item [(a)] $F_{31}(x_1, x_2, x_3)=0.15(x_1+x_2+x_3)$;
  \item [(b)] $F_{32}(x_1, x_2, x_3)=0.5x_1x_2x_3$;
  \item [(c)] $F_{33}(x_1, x_2, x_3)=\min(x_1, x_2, x_3)$;
  \item [(d)] $F_{34}(x_1, x_2, x_3)=1/\left[1+\exp\{-(x_1+x_2+x_3)\}\right]$;
  \item [(e)] $F_{35}(x_1, x_2, x_3)=\exp\{0.01x_1+0.1\sqrt{x_2}\}+\sin(x_3/5)$;
    \end{itemize}
Fifty replicate trials were generated under each model configuration, and from these the average RMSEs were calculated.  These are plotted as a function of $\sigma$ in Figure \ref{fig-3d}, again separated by initial kernel estimator (top) and initial spline estimator (bottom).   The broad patterns appear more similar to those seen in the one-predictor setting.
Table \ref{3dmse} reports the average RMSE values at $\sigma=\frac12$ and $\sigma=1$.
\begin{figure}[!ht]
\caption{Average root mean squared error (RMSE) for three-dimensional projection estimates based on initial kernel estimator (top) and initial spline estimator (bottom) for $\widetilde{F}(x)$. Digits correspond to model numbers (second digits) in Sec.~\ref{sec:simuRMSE}.}
\includegraphics[width=9cm]{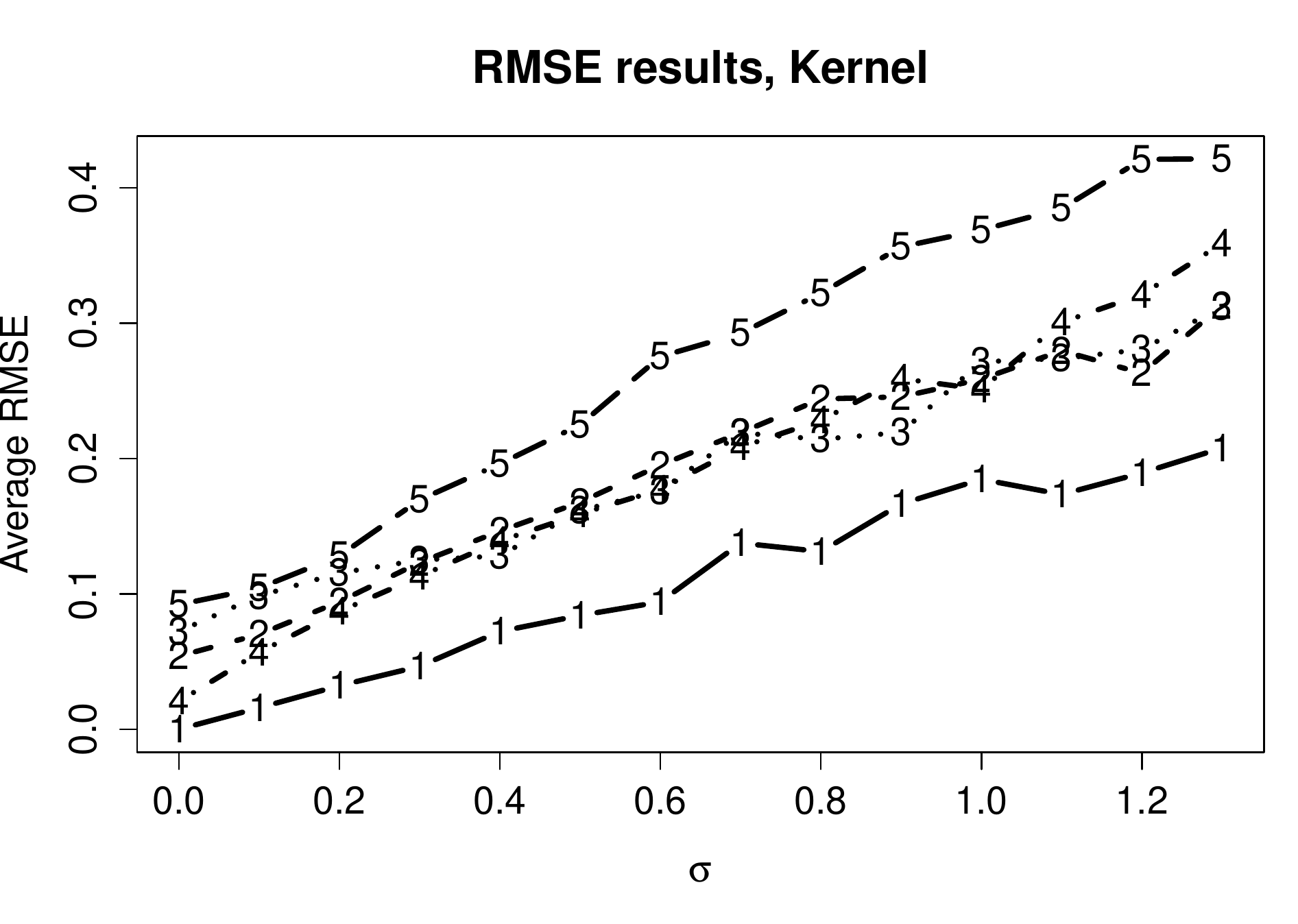}
\includegraphics[width=9cm]{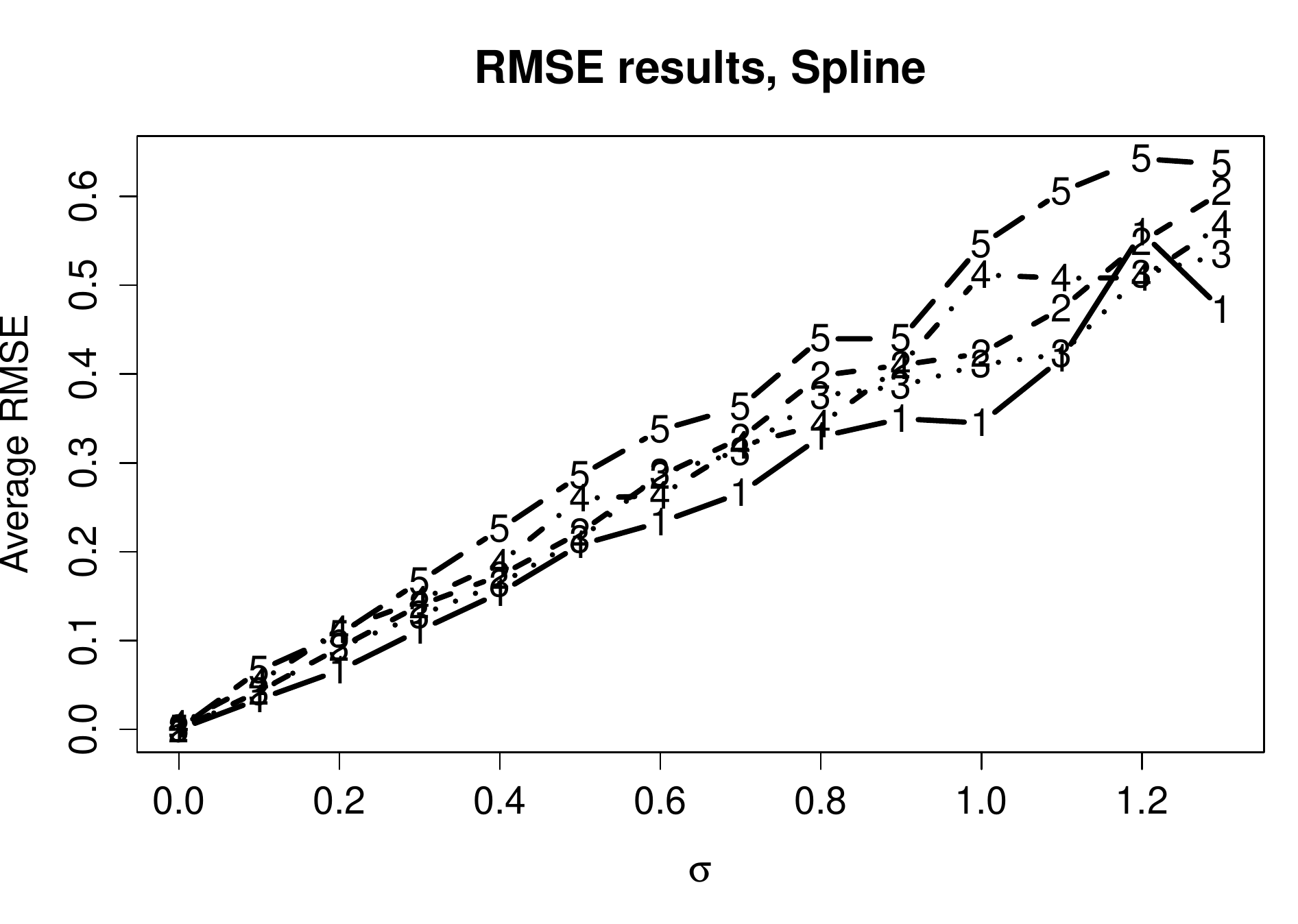}
\label{fig-3d}
\end{figure}
  \begin{table}[!h]
%
\caption{The mean and standard error of three-predictor root mean square errors (RMSEs) for simulated data at $n=100$, averaged across 50 simulation replicates, listed by underlying mean function $F(x_1, x_2, x_3)$ in  Sec.~\ref{sec:simuRMSE} and standard deviation parameter $\sigma$.}
\begin{tabular}{lllcccc}
Mean function, $F(x_1, x_2, x_3)$  & Initial estimator & & $\sigma=\frac12$ &   $\sigma=\frac12$  & $\sigma=1$ &$\sigma=1$\\
\hline
                                                 &                           &&            Average                  &    SE       &  Average         &  SE  \\
\hline
$F_{31}(x_1, x_2, x_3)$               & Kernel            &  & 0.0841  &  0.0085       & 0.1846  & 0.0147  \\
                                   & Spline            & & 0.2078  &     0.0133    & 0.3449 & 0.0262 \\ \\
$F_{32}(x_1, x_2, x_3)$               & Kernel            & & 0.1677&  0.0052         & 0.2586 & 0.0110 \\
                                   & Spline            & & 0.2216 &    0.0130      & 0.4232&  0.0216 \\ \\
$F_{33}(x_1, x_2, x_3)$               & Kernel            & & 0.1616&    0.0376       & 0.2714&  0.0464 \\
                                   & Spline            & & 0.2138&    0.0136       & 0.4108 &  0.0158\\ \\
$F_{34}(x_1, x_2, x_3)$               & Kernel            & & 0.1590&   0.0053        & 0.2511& 0.0144  \\
                                   & Spline            & & 0.2600&     0.0113      & 0.5119 &0.0175 \\ \\
$F_{35}(x_1, x_2, x_3)$               & Kernel            & & 0.2249 &       0.0065   & 0.3687& 0.0117 \\
                                   & Spline            & & 0.2857 &      0.0088    & 0.5454 & 0.0281  \\
\hline
\end{tabular}
\label{3dmse}
\end{table}

\subsection{Bootstrap interval coverage}
\label{sec-simCI}

We also used our simulation approach to explore the pointwise coverage characteristics of the bootstrap confidence intervals from Section \ref{sec:newSec3}.  Following the same procedures described in Section \ref{sec:simuRMSE}, we generated pseudo-random data $y_i \sim \text{N}\big\{F(x_{i1}),\sigma^2\big\}, \ i=1,\ldots,100$, where the one-predictor mean function $F(x_1)$ was chosen as either the sinusoidal (F$_{12}$) or logistic (F$_{16}$) form listed above. The standard deviation parameter was set to $\sigma = 1$.
The predictor variable was again taken over $x_1 \in [0,10)$.  We also generated pseudo-random two-predictor data: $y_i \sim \text{N}\big\{F(x_{i1}, x_{i2}),\sigma^2\big\}, \ i=1,\ldots,100$, where the two-predictor mean function $F(x_{i1},  x_{i2})$ was chosen as either function F$_{21}$ or F$_{26}$ from above.  As there, the predictor variables were taken over the unit square. For both the one-predictor and two-predictor settings, 2000 samples were generated at each parameter configuration.

To study pointwise coverage in the one-predictor case, we evaluated how often out of the 2000 replicate simulations the bootstrap intervals contained the true mean response at a series of values for $x_1$ over the range $0.5, 1.5, 2.5, 3.5, 5.5, \ldots, 9.5$.  We set the nominal pointwise confidence level to 95\%.  The empirical coverage rates appear in Table \ref{empircal1d}, where we see the bootstrap procedure generally contains the true mean function at or near the pointwise nominal level.  The few cases where rates drop substantively below nominal occur when the mean function turns sharply (i.e., $x=7.5$ or $x=8.5$ for the sinusoidal function $F_{12}$ and $x=3.5$ or $x=5.5$ for the logistic function $F_{16}$). In all the studies, an initial
kernel
estimate is used.
\begin{table}[!ht]
%
\caption{One-predictor empirical coverage rates ($\times 100$) based on 2000 simulation replicates, each at sample size of $n=100$.  Rates are stratified by underlying mean function $F(x_1)$, and are pointwise at each of the listed values of the predictor variable.  Nominal confidence level is 95\%.}
\begin{tabular}{lccccccccc}
                 & \multicolumn{9}{c}{$x_1$} \\
\cline{2-10}
Mean function, $F(x_1)$   & 0.5  & 1.5  & 2.5  & 3.5  & 5.5  & 6.5  & 7.5  & 8.5 & 9.5 \\
\hline
sinusoidal, $F_{12}(x_1)$ & 82.4 & 97.0 & 97.6 & 95.0 & 98.1 & 94.2 & 85.2 & 78.3 & 96.2\\
logistic, $F_{16}(x_1)$   & 98.1 & 98.5 & 90.6 & 79.8 & 75.4 & 84.7 & 86.8 & 98.1 & 99.7  \\
\hline
\end{tabular}
\label{empircal1d}
\end{table}

In similar fashion, we calculated pointwise empirical coverage with the two-predictor models over a series of $(x_1, x_2)$ pairs in the unit square.  Table \ref{empircal2d} displays the predictor pairs and the consequent coverage rates.  We find that the bootstrap procedure again generally contains the true mean function pointwise values at or near the nominal level.  Some degradation in coverage is seen at the origin $(0, 0)$, and for $F_{21}(x_1,x_2)$ also at the corner points $(0, 1)$ and $(1, 0)$.

  \begin{table}[!h]
%
\caption{Two-predictor empirical coverage rates ($\times 100$) based on 2000 simulation replicates, each at sample size of $n=100$.  Rates are stratified by underlying mean function $F(x_1)$, and are pointwise at each of the listed pairings of the predictor variables.  Nominal confidence level is 95\%.}
\begin{tabular}{clccccc}
      & & \multicolumn{5}{c}{$x_1$} \\
\cline{3-7}
$x_2$ & & $0.00$ & $0.25$ & $0.50$ & $0.75$ & $1.00$ \\
    \hline
      & & & \multicolumn{3}{c}{Mean function: $F_{21}$} \\
      \cline{4-6}
$0.00$ & & 88.5 & 96.5 & 94.1 & 92.4 & 71.1 \\
$0.25$ & & 96.6 & 98.9 & 97.5 & 88.2 & 94.0 \\
$0.50$ & & 94.2 & 97.2 & 92.4 & 95.5 & 98.9 \\
$0.75$ & & 92.4 & 88.7 & 95.3 & 94.9 & 97.3 \\
$1.00$ & & 68.9 & 93.2 & 98.2 & 96.5 & 96.8 \\
      & & & \multicolumn{3}{c}{Mean function: $F_{26}$} \\
      \cline{4-6}
$0.00$ & & 86.8 & 88.2 & 90.6 & 92.0 & 98.7 \\
$0.25$ & & 88.7 & 82.9 & 84.0 & 90.2 & 95.1 \\
$0.50$ & & 91.0 & 84.2 & 88.2 & 94.3 & 97.4 \\
$0.75$ & & 91.8 & 90.0 & 95.0 & 97.9 & 97.9 \\
$1.00$ & & 98.1 & 97.7 & 98.0 & 97.7 & 98.8 \\
\hline
\end{tabular}
\label{empircal2d}
\end{table}

\subsection{Application: environmental toxicology data}


To illustrate use of our projection method in practice, we consider two-predictor data from an environmental toxicology experiment described in \cite{Roland-walt}.  Two potentially hazardous agents, dichlorodiphenyltrichloroethane (DDT) and titanium dioxide nanoparticles (nano-TiO$_2$), were studied for their ability to induce cellular damage (as micronucleus formation) in human hepatic cells.
The two predictor variables here are taken as log-transformed concentrations of the toxins: $x_1 = \log_{10}(\text{DDT}) + 4$ and $x_2 = \log_{10}(\text{TiO}_2) + 3$.  (Control doses at zero concentrations were adjusted using consecutive-dose average spacing, from \cite{PiBa05}.)
The data in Table \ref{tab:4x4} show proportions of cells exhibiting damage after exposure to various combinations of $x_1$ and $x_2$.

\begin{table}[!ht]
\begin{center}
\caption{Proportions of human hepatic cells exhibiting micronuclei after exposure to DDT (as predictor variable $x_1$; see text) and nano-TiO$_2$ (as predictor variable $x_2$; see text); adapted from \protect\cite{Shi10}.}
\begin{tabular}{lrccccc}
\multicolumn{2}{c}{} & & \multicolumn{4}{c}{$x_2$: nano-TiO$_2$} \\
\cline{4-7}
\noalign{\smallskip}
\multicolumn{2}{c}{} & &  0 & 1 & 2& 3 \\
\noalign{\smallskip}
\hline
\noalign{\smallskip}
           & 0     & & 59/3000 & \hphantom{1}65/3000 & \hphantom{1}70/3000 & \hphantom{1}67/3000\\
           & 1 & & 67/3000 & \hphantom{1}75/3000 & \hphantom{1}83/3000 & \hphantom{1}84/3000\\[-1ex]
\raisebox{1.5ex}{$x_1$: DDT}
          & 2   & & 76/3000 & \hphantom{1}87/3000 & \hphantom{1}96/3000 & \hphantom{1}83/3000\\
          & 3    & & 94/3000 &     107/3000 &     110/3000 &     117/3000\\
\noalign{\smallskip}\hline
\end{tabular}
\end{center}
\label{tab:4x4}
\end{table}

We applied our monotonic projection method to these data in which an initial kernel estimator is used, in order to estimate the probability of response over the range of log-transformed doses.  From this, pointwise 95\% bootstrap confidence bounds were also calculated, based on 2000 bootstrap samples.  The bootstrap bounds and function estimate are plotted in Figure \ref{fig-realdata}. The display shows that the effect of TiO$_2$ is slight, but the effect of DDT is marked.  The probability of cellular damage rises quickly with exposure, then levels.  It then jumps up further after a concentration of about
0.05 $\mu$mol/L,   
such that the response appears to be essentially a step function.

\begin{figure}[!ht]
\caption{Projection estimate and pointwise 95\% bootstrap confidence bounds for data in Table \ref{tab:4x4} on proportions of human hepatic cells exhibiting micronuclei (per 3,000). The function estimate and observed data are in black, and bounds are in gray. The concentrations are given in their $\log_{10}$ scale}
\includegraphics[width = 8cm,height = 8cm]{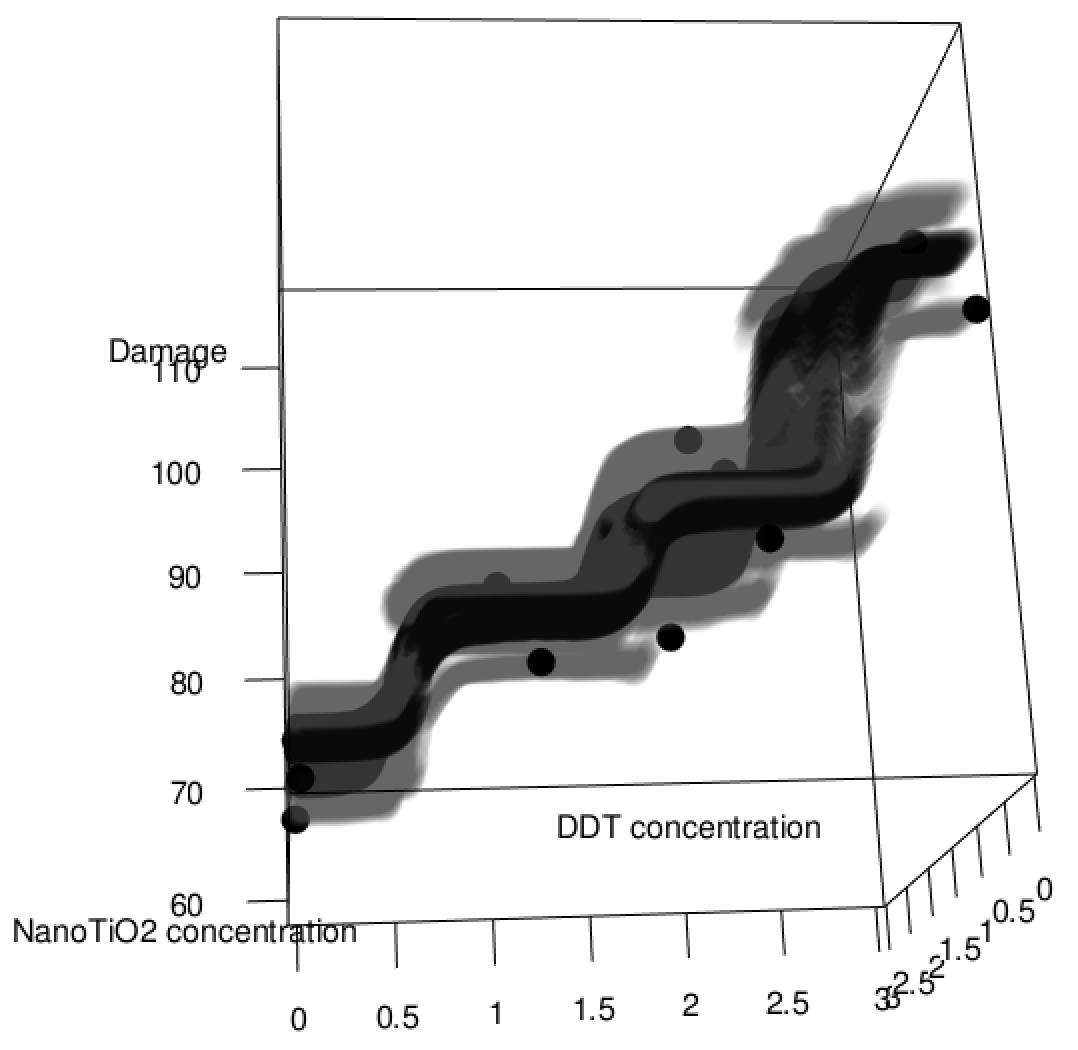}
\label{fig-realdata}
\end{figure}

\section{Discussion}
\label{sec:newSec5}
We propose a general projection framework for estimating multiple monotone regression functions.  An initial naive estimator such as the kernel or spline estimator is  first obtained, which is projected onto the monotone space serving as the ultimate estimate of the true  monotone regression regression.  The projection estimate is shown  to reduce estimation error, compared to that of an initial estimator.  Efficient computational algorithms are available for approximating the estimates. A simulation study and a data example show that the estimates possess practical finite-sample performance.

The methods exhibit good performance, although future work can expand their practicality.  For example, it is of both theoretical and practical interest to
extend the pointwise confidence bounds on the estimated surface into simultaneous confidence bands.
Constructing
confidence bands for nonparametric regression
functions
is overall a very challenging problem (see. e.g., \cite{simuband})
and we hope to report results on this under our shape constraint framework in a future work.

\section*{Appendix}

\begin{appendix}
\renewcommand{\theequation}{A.\arabic{equation}}
\renewcommand{\thelemma}{A.\arabic{lemma}}

\begin{proof} [Proof of Proposition \ref{prop-2.1}]
Our proposition falls as a special case of Lemma 2.3 of \cite{frechref}, we give another proof here which gives more insights onto the projection algorithm.

By Algorithm 1,  the multivariate projection is obtained by a collection of sequential one-dimensional projections. Let $w=\widehat{F}(x)$  be the pre-projected estimate and $P_w=\widetilde{F}(x)$  be the projection of $w$.  We first prove that Proposition \ref{prop-2.1} holds for $p=2$, the two-dimensional projection.  Therefore,
\begin{equation}
\label{eq-mother2}
 P_w =\argmin_{G \in \mathcal{M}} \int_0^1\int_0^1\{ w(s, t) - G(s, t) \}^2ds dt.
\end{equation}
  Note that $P_w$ is the limit of $\widehat{w}^{(k)}$ and $\widetilde w^{(k)}$, where $\widehat{w}^{(k)}$ is the  one-dimensional projection of $w+T^{(k-1)}$ along the $s$ direction  for any $t$, and $\widetilde w^{(k)}$ is the  one-dimensional projection of $w+S^{(k)}$ along the $t$ direction for any $s$.

Define the norm of any two-dimensional function $f(s,t)$ as
$||f||=\langle f,f\rangle^{1/2}=\left[\int \left\{f^2(s,t)\right\}ds dt\right]^{1/2}$
with $\langle\cdot,\cdot\rangle$ denoting the inner product. Then by the property of projection, one has

\begin{align}
\label{eq-important1}
\langle w-P_w,P_w \rangle=0,
\end{align}
and
\begin{align}
\label{eq-important2}
\langle w-P_w, h\rangle\leq 0\;\text{for any} \; h\in \mathcal{M}.
\end{align}

In the following, we proceed to show that for any $k$,
\begin{align}
\label{eq-important3}
 ||\widehat{w}^{(k)}||\geq||\widetilde{w}^{(k)}||\geq ||\widehat{w}^{(k+1)}||,
 \end{align}
i.e., that the norm of the sequence $\widehat{w}^{(k)}$ and $\widetilde w^{(k)}$ is not increasing.  In order to do so, we first introduce the notion of cones and dual cones of functions.
Let $C_s$ be the cone of the continuous functions $f(s,t)$ which are monotone with respect to $s$ for any $t$, and $C_t$ be the cone of continuous functions which are monotone with respect to $t$ for any $s$.
Define their dual cones $C_s^*$ and $C_t^*$ as
\begin{equation*}
C_s^*=\left\{g(s,t)\in C[0,1]^2: \int f(s,t)g(s,t)ds\leq 0,\;\text{for all}\; t \;\text{and}\; f\in C_s  \right\},
\end{equation*}
and
\begin{equation*}
C_t^*=\left\{g(s,t)\in C[0,1]^2: \int f(s,t)g(s,t)dt\leq 0,\;\text{for all}\; s \;\text{and}\; f\in C_t  \right\}.
\end{equation*}
Denote $P(w\mid C_s)$ as the projection of $w$ over $C_s$ found by minimizing $\int (w-f)^2ds$ for all $f\in C_s$ and any fixed $t$. Denote  $P(w\mid C_t)$ as the projection of $w$ over $C_t$ found by minimizing $\int (w-f)^2dt$ for all $f\in C_t$ and any fixed $s$.
 By Lemma A1 of \cite{Lin23022014}, the following holds:
\begin{equation}
\label{eq-important4}
P(w\mid C_s^*)=w-P(w\mid C_s)\;\text{and}\; P(w\mid C_t^*)=w-P(w\mid C_t).
\end{equation}
Note that  $-S^{(k+1)}=(w+T^{(k)})-\widehat{w}^{(k)}=(w+T^{(k)})-P(w+T^{(k)}|C_s)=P(w+T^{(k)}|C_s^*)$ where the last equality follows from  \eqref{eq-important4}. Here $P(w+T^{(k)}|C_s)$ denotes the projection of $w+T^{(k)}$ onto $C_s$ and $P(w+T^{(k)}|C_s^*)$ is the projection onto $C_s^*$.
Therefore  $-S^{(k+1)}$ minimizes $||(w+T^{(k)})-f||$ for all $f\in C_s^*$ and $-T^{(k)}$ minimizes $||(w+S^{(k)})-f||$ for all $f\in C_t^*$. Then, one  concludes that
$$||\widehat{w}^{(k)}||=||w+S^{(k)}-(-T^{(k-1)})||\geq ||w+S^{(k)}-(-T^{(k)})||\geq ||w+T^{(k)}-(-S^{(k+1)})||$$
for all $k$. Therefore, one has $||\widehat{w}^{(k)}||\geq||\widetilde{w}^{(k)}||\geq ||\widehat{w}^{(k+1)}||$.

Now, for sufficiently large $k$, one has
\begin{align}
\|P_w\|\leq \|\widehat{w}^{(k+1)}\|\leq \|\widetilde{w}^{(k)}\|\leq \| \widehat{w}^{(k)}\|\leq \cdots\leq \|w\|.
\end{align}
By the above equation and \eqref{eq-important2},
\begin{align*}
\|P_w-F\|^2&=\|P_w\|^2+\|F\|^2-2\langle F,P_w \rangle\\
&\leq \|w\|^2+\|F\|^2-2\langle F,w \rangle\\
&=\|w-F\|^2,
\end{align*}
proving our contention.
\end{proof}

\begin{lemma}
\label{lemma-project}
Denote $\mathcal{C}$ as the \emph{convex cone} of a function on some domain set $X=[0,1]^p$, which includes $\cM$, the convex cone of monotone functions on $X$. Let $g$ be any function on $X$ and $g^*\in \mathcal C$ such that
\begin{align}
g^*=\arg\min_{f\in \mathcal C}\int_X \big\{g(x)-f(x)\big\}^2dx.
\end{align}
Then, one has for every $f\in \mathcal C$,
\begin{equation}
\label{eq-lemeq1}
\int_{x\in X} \big\{g(x)-g^*(x)\big\}\big\{g^*(x)-f(x)\big\}dx\geq0,
\end{equation}
so that
\begin{equation}
\label{eq-lemeq2}
\int_{x\in X} \big\{g(x)-g^*(x)\big\}g^*(x)dx=0,
\end{equation}
and
\begin{equation}
\label{eq-lemeq3}
\int_{x\in X} \big\{g(x)-g^*(x)\big\}f(x)dx \leq0.
\end{equation}
\end{lemma}

\begin{proof}
By the definition of a convex cone, for any $\alpha\in [0,1]$ and any $f\in \mathcal C$, $(1-\alpha)g^*+\alpha f\in \mathcal C$.
Then
\begin{align*}
\int_X \big[g(x)-\{(1-\alpha)g^*+\alpha f \}\big]^2dx
\end{align*}
achieves its minimum at $\alpha=0$. Now, take the derivative of the above objective function to find
\begin{align*}
2\int_X \big[g(x)-\{(1-\alpha)g^*+\alpha f \} \big]\big\{g^*(x)-f(x)\big\}dx,
\end{align*}
which is non-negative at $\alpha=0$.
Therefore,
\begin{equation*}
\int_{x\in X} \big\{g(x)-g^*(x)\big\}\big\{g^*(x)-f(x)\big\}dx\geq0.
\end{equation*}
Now let $f(x)=cg^*(x)$, so that
\begin{equation*}
\int_{x\in X} \big\{g(x)-g^*(x)\big\}(1-c)g^*(x)dx\geq0.
\end{equation*}
By letting $0<c\leq 1$, for example, letting $c=1/2$, one has
\begin{equation*}
\frac{1}{2}\int_{x\in X} \big\{g(x)-g^*(x)\big\}g^*(x)dx\geq0.
\end{equation*}
Now let $c\geq 1$, e.g, $c=2$ then
\begin{equation*}
\int_{x\in X} \big\{g(x)-g^*(x)\big\}g^*(x)dx\leq0.
\end{equation*}
This implies that
\begin{equation*}
\int_{x\in X} \big\{g(x)-g^*(x)\big\}g^*(x)dx=0,
\end{equation*}
which further implies
\begin{equation*}
\int_{x\in X} \big\{g(x)-g^*(x)\big\}f(x)dx \leq0.
\end{equation*}

\end{proof}

\begin{proof}[Proof of Theorem \ref{th-1d}]

Let $\phi(u)$ be the derivative of the convex function $\Phi(u)$. By the property of convex functions, one has
\begin{align*}
\Phi(v)-\Phi(u)\geq (v-u)\phi(u).
\end{align*}
Let $u=\widetilde{F}(x)-F(x)$ and $v=\widehat{F}(x)-F(x)$.  Then
\begin{align*}
 \Phi\left\{\widehat{F}(x)-F(x)\right\}   \geq  \Phi\left\{\widetilde{F}(x)-F(x)\right\}+\left\{\widehat{F}(x)-\widetilde{F}(x)\right\}\phi\left\{\widetilde{F}(x)-F(x)\right\}.
\end{align*}
Thus,
\begin{align*}
\int_0^1 \Phi\left\{\widehat{F}(x)-F(x)\right\}dx\geq \int_0^1 \Phi\left\{\widetilde{F}(x)-F(x)\right\}dx+\int_0^1\left\{\widehat{F}(x)-\widetilde{F}(x)\right\}\phi\left\{\widetilde{F}(x)-F(x)\right\}dx.
\end{align*}
It suffices to show that
\begin{align*}
\int_0^1\left\{\widehat{F}(x)-\widetilde{F}(x)\right\}\phi\left\{\widetilde{F}(x)-F(x)\right\}dx\geq 0.
\end{align*}
Note that $\widetilde{F}(x)$ is the slope of $T(\bar{\widehat{F}}(x))$,  the greatest convex minorant of $\bar{\widehat{F}}(x)=\int_0^x \widehat{F}(s)ds.$
One can write the unit interval [0,1] as the union of the sets $\{x: T(\bar{\widehat{F}}(x))=\bar{\widehat{F}}(x)\}$,  over which the function $\widehat{F}(x)$ is monotone, and the disjoint open sets $\{x: T(\bar{\widehat{F}}(x))<\bar{\widehat{F}}(x)\}$.  One can see that over each of the disjoint sets $\{x: T(\bar{\widehat{F}}(x))<\bar{\widehat{F}}(x)\}$, $T(\bar{\widehat{F}}(x))$ is a linear function. (Otherwise, one can always construct a convex function above it, leading to a contradiction.) Therefore, $\widetilde{F}(x)$, which is the derivative of $T(\bar{\widehat{F}}(x))$, is a constant function over each of these sets.  Let  $\{x: T(\bar{\widehat{F}}(x))<\bar{\widehat{F}}(x)\}=\cup U_i$ where $U_i$ and $U_j$ are disjoint intervals for $i\neq j$.  Let $\widetilde{F}(x)=c_i$ over the set $U_i$. One has $c_i=\dfrac{\int_{U_i}\widehat{F}(x)dx }{|U_i|}$ (see Lemma 2 of \cite{giso}), where $|U_i|$ denotes the length of the intervals, which can be viewed as the projection of $\widehat{F}(x)$ restricted to the set $U_i$.  Then
\begin{align*}
\int_0^1\left\{\widehat{F}(x)-\widetilde{F}(x)\right\}\phi\left\{\widetilde{F}(x)-F(x)\right\}dx=\sum_{i}\int_{U_i} \left\{\widehat{F}(x)-\widetilde{F}(x)\right\}\phi\left\{\widetilde{F}(x)-F(x)\right\}dx.
\end{align*}

Since $F(x)$ is a monotone increasing function, $-F$ is decreasing, thus  $\phi\left\{\widetilde{F}(x)-F(x)\right\}$ is decreasing over $U_i$. Then by equation \eqref{eq-lemeq3} of Lemma \ref{lemma-project}, one has $\int_{U_i} \left\{\widehat{F}(x)-\widetilde{F}(x)\right\}\phi\left\{\widetilde{F}(x)-F(x)\right\}dx\geq 0$.  Then
\begin{align*}
\int_0^1\left\{\widehat{F}(x)-\widetilde{F}(x)\right\}\phi\left\{\widetilde{F}(x)-F(x)\right\}dx\geq 0,
\end{align*}
which implies that
\begin{equation*}
\int_0^1\Phi\left\{\widetilde{F}(x)-F(x)\right\}dx\leq \int_0^1\Phi\left\{\widehat{F}(x)-F(x)\right\}dx.
\end{equation*}

\end{proof}

\end{appendix}

\bibliography{reference-doseWP}
\bibliographystyle{abbrv}
\end{document}